\documentclass[runningheads]{llncs}
\usepackage[T1]{fontenc}
\usepackage{graphicx}
\usepackage{amsmath,amssymb}
\newcommand{\F}{\mathbb{F}}
\newcommand{\N}{\mathbb{N}}
\newcommand{\PG}{\operatorname{PG}}
\newcommand{\cP}{\mathcal{P}}
\newcommand{\cH}{\mathcal{H}}
\newcommand{\cM}{\mathcal{M}}

\begin{document}
\title{Computer classification of linear codes based on lattice point enumeration and integer linear programming}
\titlerunning{Computer classification of linear codes}
%%\author{Sascha Kurz\inst{1}\orcidID{0000-0003-4597-2041}}
\author{Sascha Kurz\orcidID{0000-0003-4597-2041}}
\authorrunning{S.~Kurz}
\institute{Friedrich-Alexander-University Erlangen-Nuremberg, Erlangen, Germany;\\ 
University of Bayreuth, Bayreuth, Germany,  
\email{sascha.kurz@uni-bayreuth.de}}
\maketitle              % typeset the header of the contribution
\begin{abstract}
Linear codes play a central role in coding theory and have applications in several branches of mathematics. For error correction purposes the minimum Hamming distance 
should be as large as possible. Linear codes related to applications in Galois Geometry often require a certain divisibility of the occurring weights. In this paper  
we present an algorithmic framework for the classification of linear codes over finite fields with restricted sets of weights. The underlying algorithms are based on lattice 
point enumeration and integer linear programming. We present new enumeration and non-existence results for projective two-weight codes, divisible codes, and additive $\F_4$-codes.

\end{abstract}

\section{Introduction}

A linear code of length $n$ is a $k$-dimensional linear subspace $C$ of the vector space $\F_q^n$, where $\F_q$ is the finite field with $q$ elements. The vectors in $C$ are called codewords. 
The weight $\operatorname{wt}(c)$ of a codeword $c=\left(c_1,\dots,c_n\right)\in C$ is the number of nonzero positions $\#\left\{ 1\le i\le n\,:\, c_i\neq 0\right\}$. For two codewords 
$c',c''\in C$ the Hamming distance is given by $d(c',c''):=\operatorname{wt}(c'-c'')$. With this, the minimum distance $d(C)$ is given by the minimum occurring Hamming distance between two different
codewords, i.e., $d(C):=\min\!\left\{d(c',c'')\,:\, c',c''\in C, c'\neq c''\right\}$. An $[n,k,d]_q$-code is a $k$-dimensional subspace $C$ of $\F_q^n$ with minimum distance at least $d$. 
We also speak of an $[n,k]_q$-code if we do not want to specify the minimum distance. Having applications in error correction in mind, a main problem in coding theory is the maximization of 
$d$, $k$, or $-n$ fixing the other two parameters. A famous problem in Galois Geometry is the maximum size of a partial $t$-spread in $\F_q^n$, where a partial $t$-spread is a set 
$\mathcal{T}$ of $t$-dimensional subspaces whose pairwise intersection is the zero vector. To such a set $\mathcal{T}$ we can associate an $[n,k]_q$-code $C$ whose weights are divisible by 
$\Delta:=q^{t-1}$, where $n=\left(q^n-1\right)/(q-1)-\#\mathcal{T}\cdot \left(q^t-1\right)/(q-1)$. So, non-existence results for $\Delta$-divisible codes imply upper bounds on the size 
of partial spreads. Using the fact that those codes have to be projective we remark that indeed all currently known upper bounds for partial spreads can be deduced from non-existence results 
for projective $\Delta$-divisible codes, see e.g\ \cite{honold2018partial} for details. Linear codes with few weights have applications in e.g.\ cryptography, designs, and secret sharing schemes. 
To sum up, there is a wide interest in the enumeration of linear $[n,k]_q$-codes with certain restrictions on the occurring weights.

Algorithms for the computer classification of linear codes date back at least to 1960 \cite{slepian1960some}, see also
\cite{betten2006error,bouyukliev2021computer,jaffe2000optimal,kaski2006classification,ostergaard2002classifying} for more recent literature. Here we want to focus on the approach using 
lattice point enumeration algorithms, see e.g.\ \cite{aardal2000solving,aardal2010lattice,schnorr1994lattice}, from \cite{bouyukliev2021computer}. We refine some of the algorithmic techniques 
and apply integer linear programming, see e.g.\ \cite{schrijver1998theory}.

The remaining part of this paper is structured as follows. We start introducing the preliminaries in Section~\ref{sec_preliminaries}. Since the geometric representation of linear codes as multisets of 
points in projective spaces plays a major role we briefly introduce this concept in Subsection~\ref{subsec_geometric_representation}. The main idea of extending linear codes using lattice point 
enumeration from \cite{bouyukliev2021computer} is briefly outlined in Subsection~\ref{subsec_extending}. The extension of each linear code is performed in two phases: In \texttt{Phase 1} lattice 
points of a certain Diophantine equation system are enumerated and in \texttt{Phase 2} additional checks are executed in order to reduce the number of extension candidates. In Section~\ref{sec_ILP} 
we state integer linear programming (ILP) formulations for some of those checks from \texttt{Phase 2} that are moved to \texttt{Phase 0} that is executed prior to \texttt{Phase 1}. Since 
we are interested in the practical performance of our proposed algorithm we discuss computational results in Section~\ref{sec_results}.   

\section{Preliminaries}
\label{sec_preliminaries}
A common representation of an $[n,k]_q$-code is as the row-span of a $k\times n$ matrix over $\F_q$ -- called generator matrix. As an example consider the 
$[56,6]_3$-code $C$ spanned by
$$
\left(\begin{smallmatrix}
%%\begin{pmatrix}
1 0 0 0 0 0 2 2 1 1 0 1 0 0 1 1 0 2 0 2 1 1 1 1 0 0 1 0 1 2 0 1 0 2 1 2 1 1 1 1 1 2 2 0 0 1 2 0 0 2 0 1 2 2 1 1 \\
0 1 0 0 0 0 1 1 1 0 1 2 1 0 1 0 1 1 2 1 1 2 0 0 1 0 2 1 1 2 2 2 1 1 1 2 1 0 0 0 0 2 1 2 0 2 2 2 0 0 2 2 2 0 1 0 \\
0 0 1 0 0 0 2 2 2 2 0 2 2 1 0 2 0 0 1 1 2 0 0 1 0 1 1 2 0 0 2 0 2 0 2 0 0 2 1 1 1 2 2 1 2 1 1 2 2 2 0 0 1 1 1 2 \\
0 0 0 1 0 0 1 0 1 1 2 2 2 2 0 2 2 1 0 2 0 0 2 2 1 0 0 1 0 1 0 1 0 0 2 2 2 2 1 0 0 2 2 2 1 1 2 1 2 2 2 2 1 2 0 0 \\
0 0 0 0 1 0 2 0 1 2 1 0 2 2 1 1 2 1 1 2 0 0 1 0 2 1 1 0 2 2 1 1 1 2 1 0 0 0 0 2 1 2 0 2 2 2 0 2 1 2 2 0 1 0 0 1 \\
0 0 0 0 0 1 1 2 2 0 2 0 0 2 2 0 1 0 1 2 1 2 2 0 0 2 0 1 1 0 2 0 1 2 1 2 2 2 2 2 1 2 0 0 2 1 0 0 2 0 2 1 1 2 2 2 \\
%%\end{pmatrix}\!.
\end{smallmatrix}\right)\!.
$$
It is the code of the famous Hill cap \cite{hill1973largest}. We say that a linear code is $\Delta$-divisible if the weights of all codewords are divisible by $\Delta$ and 
we speak of a $t$-weight code is $t$ different non-zero weights occur. Our example is a $9$-divisible $2$-weight code with weights $36$ and $45$. If a partial $3$-spread of 
$\F_3^8$ of size $248$ exists, then the set of uncovered points need to form a Hill cap, see \cite{honold2018partial} for the details.

\subsection{Geometric representation of linear codes}
\label{subsec_geometric_representation}
Permuting the columns of the stated generator matrix or multiplying arbitrary columns by arbitrary elements in $\F_q^*:=\F_q\backslash\{0\}$ yields an isomorphic linear code. We can factor 
out those symmetries by using the geometric representation of linear codes as multisets of points in a projective space $\PG(k-1,q)$. Here we just give a brief sketch of the most 
important facts and refer to e.g.\ \cite{dodunekov1998codes} for details. The $1$-dimensional subspaces of $\F_q^k$ are the points of $\PG(k-1,q)$ which we denote by $\cP_k$. A multiset 
of points $\cM$ in $\PG(k-1,q)$ is a mapping from $\cP_k$ to $\N$, i.e., to each point $P\in\cP_k$ we assign a point multiplicity $\cM(P)$. Starting from a generator matrix of a linear 
code we obtain a multiset of points by considering the multiset of one-dimensional subspaces spanned by the respective columns. In the other direction, we can take $\cM(P)$ arbitrary
generators for a point $P$ and place them at arbitrary positions in a generator matrix. Two linear codes are ismorphic iff their associated multisets of points are. We call $2$-dimensional 
subspaces lines. $(k-1)$-dimensional 
subspaces in $\PG(k-1,q)$ are called hyperplanes and there set is denoted by $\cH_k$. For each arbitrary non-zero codeword $c\in C$ also $\alpha\cdot c$ is a codeword for all 
$\alpha\in\F_q^*$ and we can associate $\F_q^*\cdot c$ with a hyperplane $H\in\cH_k$. The weight $\operatorname{wt}(c)$ of a codeword $c$ equals $\#\cM-\cM(H)$ for the associated hyperplane $H$, 
where $\#\cM:=\sum_{P\in\cP_k} \cM(P)$ and $\cM(H):= \sum_{P\in \cP_k\,:\, P\in H} \cM(P)$. The residual code of a non-zero codeword $c$ is the restriction $\cM|_H$ of $\cM$ to the corresponding 
hyperplane $H$. We say that a linear $[n,k]_q$-code $C$ is projective iff we have $\cM(P)\in\{0,1\}$ for all $P\in\cP_k$ for the corresponding multiset of points $\cM$.    

\subsection{Extending linear codes using lattice point enumeration}
\label{subsec_extending}
We say that a generator matrix $G\in\F_q^{k\times n}$ of an $[n,k]_q$-code $C$ is systematic if it is of the form $G=(I_k\vert R)$,  where $I_k$ is the $k\times k$ unit matrix and 
$R\in\F_q^{k\times (n-k)}$. Our general strategy to enumerate linear codes is to start from a systematic generator matrix $G$ of a code and to extend $G$ to a systematic generator matrix 
$G'$ of a {\lq\lq}larger{\rq\rq} code $C'$, c.f.\ \cite[Section III]{bouyukliev2021computer}. Here we assume the form
\begin{equation} \label{eq_generator_matrix}
  G'=\begin{pmatrix}
    I_k & 0\dots 0  & R \\
    0   & \underset{r}{\underbrace{1\dots 1}} & \star
  \end{pmatrix}
\end{equation}
where $G=\left(I_k|R\right)$ and $r\ge 1$, i.e., $C'$ is an $[n+r,k+1]_q$-code. Let $\cM$ and $\cM'$ be multisets of points corresponding to $C$ and $C'$, respectively. Geometrically, 
$\cM$ arises from $\cM'$ by projection through a point $P\in\cP_{k+1}$, i.e., for each line $L$ through $P$ we define $\cM(L/P)=\cM'(L)-\cM'(P)=\sum_{Q\in L\,:\, Q\neq P}\cM'(Q)$ 
and use $\PG(k,q)/P\cong \PG(k-1,q)$. Given the assumed shape of $G'$ we have $P=\left\langle e_{k+1}\right\rangle$ for the $(k+1)$th unit vector $e_{k+1}$ with $\cM'(P)=r$. 
However, we may choose any point $P\in\cP_{k+1}$ with $\cM'(P)\ge 1$ to construct $\cM$. While the linear code $C$ corresponding to $\cM$ may not admit a systematic generator matrix, there always 
is an isomorphic linear code which does. So, in general there are a lot of extensions ending up in a given code $C'$. In order to reduce the number of possible 
paths in \cite{bouyukliev2021computer} the authors speak of {\lq\lq}canonical length extension{\rq\rq} if
\begin{equation}
  \min\left\{ \cM'(Q)\,:\, \cM'(Q)>0, Q\in\cP_{k+1}\right\}=r \label{ie_min_extension}
\end{equation} 
is satisfied, c.f.\ \cite[Corollary 9]{bouyukliev2021computer}, i.e., the smallest possible value of $r$ is chosen.

The working horse for the algorithmic approach based on lattice point enumeration is \cite[Lemma 7]{bouyukliev2021computer}:
\begin{lemma}
  \label{lemma_ILP}
  Let $G$ be a systematic generator matrix of an $[n,k]_q$ code $C$ whose non-zero weights are contained in $\{i\Delta\,:a\le i\le b\}\subseteq \mathbb{N}_{\ge 1}$.
  By $c(P)$ we denote the number of columns of $G$ whose row span equals $P$ for all points $P\in\cP_k$ and set $c(\mathbf{0})=r$
  for some integer $r\ge 1$. Let $\mathcal{S}(G)$ be the set of feasible solutions of
  \begin{eqnarray}
    \Delta y_H+\sum_{P\in\mathcal{P}_{k+1}\,:\,P\le H} x_P =n-a\Delta&&\forall H\in\mathcal{H}_{k+1}\label{eq_hyperplane}\\
    \sum_{q\in\mathbb{F}_q} x_{\langle (u |q)\rangle } =c(\langle u\rangle ) && \forall \langle u\rangle \in\mathcal{P}_k \cup\{\mathbf{0}\} \label{eq_c_sum}\\
    x_{\langle e_i\rangle}\ge 1&&\forall 1\le i\le k+1\label{eq_systematic}\\
    x_P\in \mathbb N &&\forall P\in\mathcal{P}_{k+1}\label{point_var}\\
    y_H\in\{0,...,b-a\} && \forall H\in\mathcal{H}_{k+1}\label{hyperplane_var},
  \end{eqnarray}
  where $e_i$ denotes the $i$th unit vector in $\mathbb{F}_q^{k+1}$. Then, for every systematic generator matrix $G'$ of an $[n+r,k+1]_q$ code $C'$
  whose first $k$ rows coincide with $G$ and whose weights of its non-zero codewords are contained in $\{i\Delta\,:\, a\le i\le b\}$, we have a solution
  $(x,y)\in\mathcal{S}(G)$ such that $G'$ has exactly $x_P$ columns whose row span is equal
  to $P$ for each $P\in\mathcal{P}_{k+1}$.
\end{lemma}

Our algorithmic strategy is to enumerate all lattice points satisfying constraints (\ref{eq_hyperplane})-(\ref{hyperplane_var}) in \texttt{Phase~1} and consider them as extension candidates $C'$ 
for a given linear $[n,k]_q$-code $C$, where additional checks may be applied, in \texttt{Phase~2}. Of course we have to deal with the problem of eliminating isomorphic copies. On the other hand, there
are there are some theoretic insights that allow to directly reject some of the lattice points as candidates in \texttt{Phase~2}, see \cite{bouyukliev2021computer} for details. 

For our purpose, a few remarks are in order. Equations~(\ref{eq_hyperplane}) ensure that $C'$ is $\Delta$-divisible with minimum weight at least $a\Delta$ and maximum weight at most $b\Delta$.  
Of course we may always choose $\Delta=1$, but the larger we choose $\Delta$ and the tighter we choose $a$, $b$ the less lattice points will satisfy constraints (\ref{eq_hyperplane})-(\ref{hyperplane_var}). 
Inequalities~(\ref{eq_c_sum}) and (\ref{eq_systematic}) model the assumed shape of (\ref{eq_generator_matrix}). (Technically, Inequalities (\ref{eq_systematic}) are removed in a
preprocessing step before calling a software for the enumeration of lattice points.) The variables $x_P$ model $\cM'(P)$ and the variables $y_H$ parameterize $\cM'(H)$ as detailed in (\ref{eq_hyperplane}). 
Actually, constraints (\ref{point_var}) and (\ref{hyperplane_var}) are just saying that we are only interested in lattice points, i.e., integral solutions.
 
Due to the availability of practically fast lattice point enumerations algorithms, the algorithmic strategy to generate many extension candidates in \texttt{Phase 1} and filter out suitable 
candidates afterwards in \texttt{Phase 2} turned out to be quite efficient if the number of constraints and variables is not too large, see \cite{bouyukliev2021computer}. It was also observed that 
considering just a subset of the constraints (\ref{eq_hyperplane}) can reduce computation times in many situations, i.e., generating more candidates in \texttt{Phase 1} can pay off if a simpler system 
allows faster generation of lattice points and the checks in \texttt{Phase 2} can be implemented efficiently.

In the implementation described in \cite{bouyukliev2021computer}, the check of condition (\ref{ie_min_extension}) as well as checks based on possible gaps in the assumed weight spectrum      
$\{i\Delta\,:a\le i\le b\}$ are moved to \texttt{Phase 2}. The idea of this paper is to demonstrate that it sometimes can pay off to move such checks to integer linear programming computations 
in a \texttt{Phase 0} prior to \texttt{Phase 1}.

\section{Enhancing the algorithmic approach via integer linear programming computations}
\label{sec_ILP}
Of course one can check the feasibility or infeasibility of an ILP using lattice point enumeration, after possibly transforming inequalities into equalities. In the other 
direction, many ILP solvers can also enumerate all lattice points of a polytope. However, in many situations ILP solvers can find a single feasible solution or show infeasibility faster 
than the full enumeration by a lattice point enumeration algorithm. However, the situation changes if one wants to enumerate a larger set of solutions exhaustively. 

So, a first idea is to check feasibility of constraints (\ref{eq_hyperplane})-(\ref{hyperplane_var}) using an ILP solver in a \texttt{Phase 0} prior to \texttt{Phase 1}. This pays off in those situations 
where some extension problems don't have a solution. We remark that several calls of ILP solvers for different random target functions can also be used in a heuristic approach if it is not necessary 
to classify all possible codes with certain parameters but just to find some examples.

The second idea is to move some of the checks of \texttt{Phase 2} into an initial feasibility test based on ILP computations. Starting with 
Inequality~(\ref{ie_min_extension}) we observe that it can be rewritten as

\begin{equation}
  x_P=0\quad\vee\quad x_P\ge r \label{ie_min_extension_2}
\end{equation}
  
for every point $P\in\cP_{k+1}$. Having an upper bound $x_P\le \Lambda_P$ at hand, which we have in most applications, we can linearize Inequality~(\ref{ie_min_extension_2}), as
\begin{equation} 
  x_P\le \Lambda_P \cdot u_P \quad\wedge\quad x_P\ge r\cdot u_P  \label{ie_min_extension_3}
\end{equation}
using an additional binary variable $u_P$.

Instead of solving a larger ILP some cases can also be eliminated in a simple preprocessing step. Suppose that an instance of Inequality~(\ref{eq_c_sum})
reads $$
  \sum_{i=1}^q x_{P_i} =c
$$
and we have $x_{P_i}\le\Lambda$ as well as $x_{P_i}=0\,\,\vee\,\, x_{P_i}\ge r$ for all $1\le i\le r$. If
\begin{equation}
  \left\lfloor\frac{c}{r}\right\rfloor < \left\lceil\frac{c}{\Lambda}\right\rceil,
\end{equation}
then no solution exists since at most $\left\lfloor\frac{c}{r}\right\rfloor$ variables $x_{P_i}$ have to be non-zero and at least $\left\lceil\frac{c}{\Lambda}\right\rceil$ variables 
$x_{P_i}$ have to be non-zero. We have observed the applicability of this criterion in practice for parameters $(q,r,\Lambda,c)=(2,3,4,5)$.

With respect to gaps in the weight spectrum we assume that the possible non-zero weights are contained in 
\begin{equation}
  \left\{a_1\Delta_1,\dots,b_1\Delta_1,a_2\Delta_2,\dots,b_2\Delta_2,\dots,a_l\Delta_l,\dots,b_l\Delta_l\right\},\label{weight_spectrum}
\end{equation}
where $l\ge 2$, $a_i<b_i$ for $1\le i\le l$, and $b_{i-1}\Delta_{i-1}<a_i\Delta_i$ for $2\le i\le l$. With this, we can replace 
Inequalities~(\ref{eq_hyperplane}) and Inequalities~(\ref{hyperplane_var}) by
\begin{eqnarray}
  \sum_{i=1}^l \Delta_i y_H^i+\!\!\!\!\sum_{P\in\mathcal{P}_{k+1}\,:\,P\le H}\!\!\!\! x_P =n-\sum_{i=1}^l a_i\Delta_i z_H^i \Delta && \forall H\in\mathcal{H}_{k+1}, \label{ie_add_first}\\ 
  y_H^i \le \left(b_i-a_i\right) z_H^i && \forall H\in\mathcal{H}_{k+1}, \forall 1\le i\le l, \\ 
  \sum_{i=1}^l z_H^i = 1 && \forall H\in\mathcal{H}_{k+1},\\ 
  z_H^i \in \{0,1\} && \forall H\in\mathcal{H}_{k+1}, \forall 1\le i\le l, \\
  y_H^i \in \N && \forall H\in\mathcal{H}_{k+1}, \forall 1\le i\le l \label{ie_add_last}. 
\end{eqnarray} 

The ILP for \texttt{Phase 0} consists of inequalities (\ref{eq_c_sum})-(\ref{point_var}) and (\ref{ie_add_first})-(\ref{ie_add_last}) if we 
have a possible weight spectrum as in (\ref{weight_spectrum}) with $l\ge 2$ or, alternatively, of inequalities (\ref{eq_hyperplane})-(\ref{hyperplane_var}).      
If $r\ge 2$, see (\ref{eq_generator_matrix}), then we additionally add the constraints (\ref{ie_min_extension_3}) and additional variables $u_P\in\{0,1\}$ for 
all $P\in\cP_{k+1}$.

\section{Computational results}
\label{sec_results}

In this section we want to present a few computational results that have been obtained with our refined algorithmic approach of computer classification of linear 
codes based on lattice point enumeration and integer linear programming. We start with non-existence results for projective two-weight codes, see e.g.\ 
\cite{brouwer2021two,calderbank1986geometry} for surveys. To this end, we slightly modify the notion of an $[n,k,d]_q$-code by replacing $d$ with a set of occurring
non-zero weights and writing $\le n$ if the length is at most $n$. For each projective two-weight code over $\F_q$ with characteristic $p$ there exist integers $u$ 
and $t$ such that the two occurring non-zero weights can be written as $u\cdot p^t$ and $(u+1)\cdot p^t$, i.e.\ the code is $p^t$-divisible, see \cite{delsarte1972weights}.
Under some mild extra conditions this is also true for non-projective two-weight codes \cite{kurz2024non}. 

\begin{proposition}
  \label{lemma_no_66_5_4_2wt_code}
  No projective $[66,5,\{48,56\}]_4$-code exists.
\end{proposition}
\begin{proof}
  By exhaustive enumeration we have determined all three non-isomorphic $[\le 65,3,\{48,56\}]_4$-codes. 
  %% Using \texttt{LinCode} \cite{bouyukliev2021computer} we have enumerated all $[\le 65,3,\{48,56\}]_4$-codes:
  %% $$
  %%   \begin{pmatrix}
  %%     111111111111111111111111111111111111111111111110000000000010000\\
  %%     000000000001111111111112222222222223333333333331111111111101000\\
  %%     001112223330001112223330001112223330001112223330011122233300111    
  %%   \end{pmatrix}\!,
  %% $$
  %% $$
  %%   \begin{pmatrix}
  %%     1111111111111111111111111111111111111111111111111111111000001000\\
  %%     0000000000000111111111111112222222222222233333333333333111110100\\
  %%     0001111223333001111222233330000112222333300001111222233011330011    
  %%   \end{pmatrix}\!,
  %% $$
  %% $$
  %%   \begin{pmatrix}
  %%     11111111111111111111111111111111111111111111111111111110000000100\\
  %%     00000000000000011111111111111112222222233333333333333331111111010\\
  %%     00111222223333300000111112333330111222300000111222223330011123001    
  %%   \end{pmatrix}\!.
  %% $$ 
  They have lengths $63$, $64$, $65$ and orders $362880$, $1728$, $36$ of their automorphism groups, respectively. None of them can be extended to an $[65,4,\{48,56\}]_4$-code, 
  so that no projective $[66,5,\{48,56\}]_4$-code exists.\hfill{$\square$}  
\end{proof}

\begin{proposition}
  \label{lemma_no_35_4_8_2wt_code}
  No projective $[35,4,\{28,32\}]_8$-code exists.
\end{proposition}
\begin{proof}
  We have computationally determined the unique $[\le 34,3,\{28,32\}]_8$-code. 
  %% $$
  %%   \begin{pmatrix}
  %%     1111111111111111111111111111111100\\
  %%     0000111112333334444455555677777010\\
  %%     2355233770002452244745577601334001    
  %%   \end{pmatrix}\!.
  %% $$  
  It has length $34$ and an automorphism group of order $43008$. We have computationally checked that no extension exists. The corresponding ILP computation 
  took 5.6~hours of computation time and checked 1\,633\,887 B\&B-nodes.\hfill{$\square$}
\end{proof}

We remark that we have also enumerated all $[\le 122,3,\{108,117\}]_9$-codes (with maximum point multiplicity $9$). All of these $1147$ non-isomorphic 
codes have length $n=122$. It is an interesting open question whether one of these can be extended to a projective $[123,4,\{108,117\}]_9$-code, which 
is currently unknown. Applying the so-called subfield construction, see e.g.\ \cite{calderbank1986geometry}, would also yield a projective 
$[492,8,\{324,351\}]_3$-code -- again currently unknown.

On the way to enumerate all projective $[77,5,\{56, 64\}]_4$-codes we need to consider the extension of $[74,3,\{56, 64\}]_4$- to $[76,4,\{56, 64\}]_4$-codes. 
Without the application of \texttt{Phase 0} exactly 1\,087\,803 linear codes are constructed from lattice points and eliminated by Inequality~(\ref{ie_min_extension})
afterwards. We remark that there are exactly $5$ nonisomorphic $[76,4,\{56, 64\}]_4$-codes.      

\medskip

Next we consider projective $\Delta$-divisible codes. The characterization of all possible 
length $n$ such that there exists a projective $8$-divisible $[n,k]_2$-code for some $k$ 
was completed in \cite{honold2019lengths}. For projective $5$-divisible codes over $\F_5$ the 
only undecided length is $n=40$, see e.g.\ \cite{kurz2021divisible} for a survey.
\begin{proposition}
  No projective $5$-divisible $[40,4]_5$-code exists.
\end{proposition}
\begin{proof}
  By exhaustive enumeration we have determined all $371$ non-isomorphic $5$-divisible $[39,3]_5$-codes 
  with maximum point multiplicity $4$. None of them can be extended to a projective $5$-divisible $[40,4]_5$-code. 
  Note that extending a $[39,3]_5$-code with maximum point multiplicity to a projective $[40,4]_5$-code would imply 
  that the resulting code contains a two-dimensional simplex code, in geometrical terms a line, in its support. 
  However, no projective $5$-divisible $[34,k]_5$-code exists \cite[Lemma 7.12]{kurz2021divisible}.\hfill{$\square$} 
\end{proof}
%% generate 5-div [39,3]_5 codes with maximum point multiplicity 4: [otherwise we have a line]
%% 
%% DIFFERENT STRUCTURES :9   ALL CODES: 371
%% 
%%     1 61 $1z^{0}+12z^{25}+64z^{30}+48z^{35} GL-K:0 GL-LK:0$\\
%%     2 16 $1z^{0}+88z^{30}+36z^{35} GL-K:0 GL-LK:0$\\
%%     3 106 $1z^{0}+4z^{25}+80z^{30}+40z^{35} GL-K:0 GL-LK:0$\\
%%     4 173 $1z^{0}+8z^{25}+72z^{30}+44z^{35} GL-K:0 GL-LK:0$\\
%%     5 5 $1z^{0}+4z^{20}+4z^{25}+68z^{30}+48z^{35} GL-K:0 GL-LK:0$\\
%%     6 2 $1z^{0}+4z^{20}+76z^{30}+44z^{35} GL-K:0 GL-LK:0$\\
%%     7 1 $1z^{0}+8z^{20}+64z^{30}+52z^{35} GL-K:0 GL-LK:0$\\
%%     8 5 $1z^{0}+16z^{25}+56z^{30}+52z^{35} GL-K:0 GL-LK:0$\\
%%     9 2 $1z^{0}+20z^{25}+48z^{30}+56z^{35} GL-K:0 GL-LK:0$\\
%% 
%% Extension to projective 5-dive [40,4]_5 codes in 37h excluded

\medskip

As a last example we consider additive codes over $\F_4$. In general, each $k$-dimensional additive code of length $n$ with minimum
Hamming distance $d$ over $\F_{q^2}$ geometrically corresponds to a multiset of lines in $\PG(2k-1, q)$ such that
each hyperplanes contains at most $n-d$ lines, see e.g.\ \cite{blokhuis2004small}. For $\F_4$ the case $k=3.5$ was partially 
studied in \cite{guan2023some}. Our aim is to construct examples for $(n,d)=(51,38)$ which are special $[153,7,76]_2$-codes. 
\begin{lemma} 
  Let $\cM$ be the multiset of points in $\PG(6,2)$ formed by $51$ lines such that every hyperplane contains at most $13$ lines and $C$ denote the 
  corresponding binary code. Then, $C$ is $[153,7,76]_2$-code with $\operatorname{wt}(c)\le 102$ for every $c\in C$.    
\end{lemma}
\begin{proof}
  Replacing the $51$ lines by their $3$ points yields a multiset $\cM$ in $\PG(6,2)$ with cardinality $153$. Since each hyperplane $H$ contains between $0$ and $13$ lines, 
  we have $51\le\cM(H)\le 77$, so that $76\le \operatorname{wt}(c)\le 102$ for all $c\in C\backslash\{\mathbf{0}\}$.\hfill{$\square$}  
\end{proof}

\begin{lemma}
  No projective $2$-divisible $[65,6,32]_2$-code with maximum point multiplicity exists.
\end{lemma}
\begin{proof}
  By exhaustive enumeration.\hfill{$\square$}
\end{proof}
We remark that there exists a unique $[63,6,32]_2$-code $C$ which is projective and $32$-divisible. If $\cM$ denotes the corresponding multiset of 
points, then adding two arbitrary points yields $[65,6,32]_2$-codes.

\begin{lemma}
  \label{weight_restrictions}
  Let $C$ be a $[153,7,76]_2$-code with $\operatorname{wt}(c)\le 102$ for every $c\in C$. Then, the occurring non-zero weights are contained in $\{76,80,92,96,100\}$.
\end{lemma}
\begin{proof}
  Note that $C$ is a Griesmer code and that $4$ divides the minimum distance $76$, so that we can conclude that $C$ is $4$-divisible \cite{ward1998divisibility}. 
  The residual code of a codeword of weight $84$ would be a $[69,6,34]_2$-code which does not exist. Now let $C'$ be the residual code of a codeword of 
  weight $88$, so that $C'$ is $[65,6,32]_2$-code. Since $C$ is $4$-divisible and has maximum point multiplicity $2$, $C'$ has to be $2$-divisible with maximum point 
  multiplicity $2$. However, we have just shown that such a code does not exist.\hfill{$\square$}    
\end{proof}

\begin{proposition}
  There are exactly two non-isomorphic $[153,7,76]_2$-codes $C$ with $\operatorname{wt}(c)\le 102$ for every $c\in C$.
\end{proposition}
\begin{proof}
  %% We have exhaustively enumerated all $[151,6]_2$-codes with weight restrictions as in Lemma~\ref{weight_restrictions}.
  We have exhaustively enumerated all $[\le 150,5]_2$-codes with weight restrictions as in Lemma~\ref{weight_restrictions}. There are $5$ such codes with length $149$ and $27$ such 
  codes with length $150$. Extending them to $[151,6]_2$-codes with weight restrictions as in Lemma~\ref{weight_restrictions} leaves just two possibilities (with maximum point 
  multiplicity $4$). 
  Noting that $C$ has maximum point multiplicity $2$, it has to be an extension of such a $[151,6]_2$-code and we have computationally verified the stated results. 
  Those two codes have weight distributions $76^{107} 80^{15} 92^{5}$ and $76^{108} 80^{14} 92^{4} 96^1$. The automorphism groups have orders $16128$ and $32256$, respectively.\hfill{$\square$}       
\end{proof}

A generator matrix of the $[153,7,\{76,80,92\}]_2$-code with an automorphism group of order $16128$ is given by the concatenation of 
$$
\left(
\begin{smallmatrix}
1111111111111111111111111111111111111111111111111111111111111111111111111111\\
0000000000000000000000000000000000000001111111111111111111111111111111111111\\
0000000000000000000111111111111111111110000000000000000000011111111111111111\\
0000000000011111111000000000000111111110000000000001111111100000000000011111\\
0000011111100001111000000111111000011110000001111110000111100000011111100001\\
0011100011100110011000111000111001100110001110001110011001100011100011100110\\
0100100101101010100011011011001010101110110010110110111010100101100100101000
\end{smallmatrix}\right)
$$
and
$$
\left(
\begin{smallmatrix}
11100000000000000000000000000000000000000000000000000000000000000000010000000\\ 
11111111111111111111111111111111111111111100000000000000000000000000001000000\\
11100000000000000000001111111111111111111111111111111111111110000000000100000\\
11100000001111111111110000000011111111111100000001111111111111111111000010000\\
11100011110000001111110000111100000011111100011110000001111110001111100001000\\
01101100110001110001110011001100011100011101100110001110001110110011100000100\\
10110100010010010110010101110101101100101100101010110010010011010111100000011
\end{smallmatrix}\right)\!\!.
$$
The corresponding multisets of points can be partitioned into $51$ lines. We remark that there is also 
an $[153,7,76]_2$-code with weight distribution $76^{108} 80^{15} 92^3 108^1$ which cannot be partitioned in such a way due to the codeword of weight $108$. 
Without excluding the weights $84$ and $88$ the intermediate codes would be too numerous to perform an exhaustive enumeration in reasonable time. 

\end{document}